\documentclass[12pt]{article}

\usepackage{color}
\usepackage{amsmath, amssymb} 
\usepackage{subfig}  
\usepackage{hyperref}   
\usepackage{aliascnt}  
\usepackage{caption}
\usepackage{float}
\usepackage{xspace}

\usepackage{multirow}
\usepackage{makecell}
\usepackage{algorithm}
\usepackage{algpseudocode}

\newcommand{\cut}[1]{}

\allowdisplaybreaks 

\usepackage{aliascnt}  		

\newtheorem{theorem}{Theorem}          	
\newaliascnt{lemma}{theorem}				
\aliascntresetthe{lemma}  					
\newaliascnt{conjecture}{theorem}			
\aliascntresetthe{conjecture}  				
\newaliascnt{remark}{theorem}				
\aliascntresetthe{remark}  					
\newaliascnt{corollary}{theorem}			
\aliascntresetthe{corollary}  				
\newaliascnt{definition}{theorem}			
\newtheorem{definition}[definition]{Definition}    
\aliascntresetthe{definition}  				
\newaliascnt{proposition}{theorem}			
\aliascntresetthe{proposition}  				
\newaliascnt{example}{theorem}			
\aliascntresetthe{example}  				
\newtheorem{claim}{Claim}

\newenvironment{proof}{\par\noindent {\em Proof.}\ } {\hfill$\Box$\par\medskip}

\providecommand{\E}[0]{\mathsf{E}}

\newcommand{\set}[1]{\{#1\}}                    

\newcommand{\R}{\mathsf{R}}
\newcommand{\Cl}{\mathsf{Cl}}
\newcommand{\new}{\mathsf{new}}

\begin{document}

\title{A Note on the Hardness of \\ the Critical Tuple Problem}
\author{Egor V.~Kostylev \\ \small{University of Oxford,} \\ \small{Oxford, UK} \\ \small{egor.kostylev@cs.ox.ac.uk} \and Dan Suciu \\ \small{University of Washington,} \\ \small{Seattle, WA, USA} \\ \small{suciu@cs.washington.edu}}

\maketitle

\section{Introduction}

The notion of critical tuple was introduced by Miklau and
Suciu~\cite{DBLP:journals/jcss/MiklauS07}, who also claimed that the problem of checking
whether a tuple is non-critical is $\Pi^p_2$-complete.  Kostylev identified
an error in the 12-page-long hardness proof.  It turns out that the issue is rather fundamental:
the proof can be adapted to show hardness of a relative variant of
tuple-non-criticality, but we have neither been able to prove the original
claim nor found an algorithm for it of lower complexity.  In this note we state formally the relative variant and present an alternative, 
simplified proof of its $\Pi^p_2$-hardness; we also give an NP-hardness proof for the original problem, the best lower bound we have been able to show.  Hence, the precise complexity of the original critical tuple problem remains open.

\section{Problem definition}

Consider a first-order signature without functional symbols.
A \emph{tuple} (or \emph{fact}) is a ground atom.
A {\em database instance}, or \emph{instance}
for short, is a set of tuples.
A {\em Boolean conjunctive query}, or just \emph{query} for short, is an implicitly existentially quantified first-order
sentence of the form
\begin{align*}
  \R_1(\bar x_1) \wedge \dots \wedge \R_m(\bar x_m)
\end{align*}
over \emph{atoms} $\R_i(\bar x_i)$, for $i = 1, \ldots, m$, with predicates $\R_i$ and vectors of variables $\bar x_i$. 
A {\em homomorphism} from a query $Q$ to an instance $I$ is a mapping from the variables of $Q$ to
constants such that every atom of $Q$ is mapped to a tuple in $I$. It is well-known that if there exists such a homomorphism, then $Q$ holds on $I$.

The following is the definition of a critical tuple as it is given in~\cite{DBLP:journals/jcss/MiklauS07} (modulo straightforward equivalent modifications).

\begin{definition}
A tuple $\tau$ is \emph{critical} for a query $Q$ if and only if there exists an instance $I$ with a homomorphism from $Q$ to $I$ such that there is no homomorphism from $Q$ to $I \setminus \{\tau\}$.
\end{definition} 

Next we give the definition of a relative critical tuple, which is slightly different from the previous one. However, as we will see later, this difference may have important consequences for the computational complexity of the considered problems.

\begin{definition}
A tuple $\tau$ is \emph{critical} for a query $Q$ \emph{relative to} an atom $g$ in $Q$ if and only if there exists an instance $I$ with a homomorphism from $Q$ to $I$ sending $g$ to $\tau$ such that there is no homomorphism from $Q$ to $I \setminus \{\tau\}$.
\end{definition}


We note that if $\tau$ is critical for $Q$ relative to an atom $g$ then it is critical for $Q$.
However, the converse fails. Indeed, take
$Q = \R(x,y,z,z) \land \R(x,x,y,y)$ over variables $x$, $y$ and $z$, and $\tau=\R(a,a,b,b)$ over constants $a$ and $b$. Then $\tau$ is
non-critical relative to the first atom, $g_1 = \R(x,y,z,z)$, because if
a homomorphism $h$ from $Q$ to an instance $I$ maps $g_1$ to $\tau$ then it sets $h(x)=h(y)=a$, and
therefore the image of the second atom is the tuple $\R(a,a,a,a) \in I$
making $\tau$ redundant for witnessing $g$.  However, $\tau$ is critical for $Q$
because $Q$ holds on the instance $I=\set{\R(a,b,c,c),\R(a,a,b,b)}$,
but false on $I \setminus \set{\tau}= \set{\R(a,b,c,c)}$.

As mentioned above, Miklau and Suciu showed $\Pi_2^p$-completeness of the problem of checking whether a tuple is not critical for a query~\cite{DBLP:journals/jcss/MiklauS07}. Alas, the proof of hardness is incorrect. The problem is that the constructed mapping $h_{new}$ from the variables of the query $Q$ to the instance $I - \{t\}$ (i.e., $I \setminus \{\tau\}$ in our notation) may not be a homomorphism, contrary to what is claimed. In particular, in case $A=X_2$ in Section A.13 when $(I, h)$ is a bad instance and $h(x_2) = 0$ we have that $h_{new}(x_2) = h(x_2^b)$, while $h_{new}(e_i)
= h(e_i)$ for $i = 1, 2$, since $h_{new}(z) = h(z)$. Hence, $h_{new}$ is not
a homomorphism, because it does not send atom $g_{X_2 = false}$ anywhere---$closure(g_{X_2 = false})$ does not have $h_{new}(g_{X_2 = false}) = N(a, z^b,
x_2^b, z, e_1, e_2)$, and $h_{new}(g_{X_2 = false})$ is not $g_{X_2 = false}$ by
itself.

\section{\!\!\!Complexity of relative critical tuple problem}

In this section we prove the $\Pi_2^p$-completeness of the problem of checking whether a tuple is not critical for a query relative to an atom. Note, however, that this result immediately implies only $\Pi_2^p$-membership of the original critical tuple problem, but says nothing about its hardness.

\begin{theorem}
Given a tuple $\tau$, a query $Q$, and an atom $g$ in $Q$, it is $\Pi^p_2$-complete to decide whether $\tau$ is not critical for $Q$ relative to $g$.
\label{thm_rct}
\end{theorem}

\begin{proof}
The upper complexity bound is straightforward, because we can restrict the search for an instance witnessing criticality to  the search only through the instances of the size at most the size of the query $Q$; therefore, a possible $\Pi^p_2$ algorithm first guesses such an instance $I$ with a homomorphism from $Q$ to $I$ sending $g$ to $\tau$, and then calls for an NP oracle to check that there is no homomorphism from $Q$ to $I \setminus \{\tau\}$.

Next we prove the lower complexity bound. The proof is by reduction of $\forall\exists$3SAT---that is, of a canonical $\Pi_2^p$-complete problem. Consider a quantified Boolean formula $\phi = \forall \bar u.\, \exists \bar v.\, \psi$, where $\bar u$ and $\bar v$ are vectors of distinct propositional variables, and $\psi$ is a conjunction of clauses of the form $\alpha_1 \lor \alpha_2 \lor \alpha_3$ with each $\alpha_i$ either a propositional variable in $\bar u \cup \bar v$ or the negation of such a variable. We assume that each variable $u \in \bar u$ appears in $\psi$ both positively and negatively (we can do it without loss of generality, because if it is not the case for some $u$, then we can add to $\psi$ a clause $u \lor u \lor \neg u$ without changing the validity of $\phi$).
We construct a tuple $\tau$, a query $Q$, and an atom $g$ in $Q$ such that $\tau$ is not critical for $Q$ relative to $g$ if and only if $\phi$ is valid.

\smallskip

In the reduction we will use a predicate $\R$ of arity 4 and, for each clause $\gamma$ in $\psi$, a predicate $\Cl_\gamma$ of arity 6.

First, let $\tau$ be the tuple $\R(0, 0, 1, 2)$, where $0$, $1$ and $2$ are constants. 

Then, let $Q$ consist of the following, where all arguments are variables:
\begin{itemize}
\item[--] the atoms $\R(z, z', y, y')$ and $\R(s, s', p, p)$;
\item[--] the atoms $\R(x_u, x_u, y_u, y'_u)$, $\R(x_{\neg u}, x_{\neg u}, y'_{\neg u}, y_u)$, and $\R(f, f, y_u, y_u)$ for each $u \in \bar u$;
\item[--] for every clause $\gamma = \alpha_1 \lor \alpha_2 \lor \alpha_3$ in $\psi$ with each literal $\alpha_i$ over a propositional variable $w_i$, 
\begin{itemize}

\item[-] the atom 
\begin{equation}
\Cl_\gamma(x_{1}, x_{2}, x_{3}, r, r, z), \label{eq_clauses}
\end{equation}
where, for $i = 1, 2, 3$,
$$ 
x_i = \left\{\begin{array}{ll} x_{\alpha_i}, & \text{ if } w_i \in \bar u, \\ x_{w_i}, & \text{ if } w_i \in \bar v; \end{array}\right.
$$

\item[-] the atoms
\begin{equation}
\Cl_\gamma(b^\pi_{\alpha_1}, b^\pi_{\alpha_2}, b^\pi_{\alpha_3}, z, z', s),
\label{eq_assignments}
\end{equation}
for each (of 7) assignment $\pi$ of $w_1$, $w_2$ and $w_3$ satisfying $\gamma$, where, for $i = 1, 2, 3$, 
$$
b^\pi_{\alpha_i} = \left\{\begin{array}{ll} x_{\alpha_i}, & \text{ if } \pi(\alpha_i) = \texttt{true} \text{ and } w_i \in \bar u, \\ 
t, & \text{ if } \pi(w_i) = \texttt{true} \text{ and } w_i \in \bar v, \\ f, & \text{ otherwise};
\end{array}\right.
$$
\item[-] and the atoms
\begin{equation}
\Cl_\gamma(b^\pi_{\alpha_1}, b^\pi_{\alpha_2}, b^\pi_{\alpha_3}, s, s', s),
\label{eq_backups}
\end{equation}
for each (of 8) assignment $\pi$ of $w_1$, $w_2$ and $w_3$, where 
the $b^\pi_{\alpha_i}$ are as before.
\end{itemize}
\end{itemize} 

For example, for a clause $\gamma = u_1 \lor \neg u_2 \lor \neg v$ with $u_1, u_2 \in \bar u$ and $v \in \bar v$ query $Q$ contains atoms
\begin{equation}
\begin{array}{lrrrrrrrllrrrrrrr}
\Cl_\gamma(& x_{u_1}, & x_{\neg u_2}, & x_{v}, & r, & r, & z & ), \\ \\
\Cl_\gamma(& f, & f, & f, & z, & z', & s & ),  & & \Cl_\gamma(& f, & f, & f, & s, & s', & s & ),\\
& & & & & & & & & \Cl_\gamma(& f, & f, & t, & s, & s', & s &), \\
\Cl_\gamma(& f, & x_{\neg u_2}, & f, & z, & z', & s & ), & & \Cl_\gamma(& f, & x_{\neg u_2}, & f, & s, & s', & s & ), \\
\Cl_\gamma(& f, & x_{\neg u_2}, & t, & z, & z', & s &), & & \Cl_\gamma(& f, & x_{\neg u_2}, & t, & s, & s', & s & ), \\
\Cl_\gamma(& x_{u_1}, & f, & f, & z, & z', & s & ), & & \Cl_\gamma(& x_{u_1}, & f, & f, & s, & s', & s & ), \\
\Cl_\gamma(& x_{u_1}, & f, & t, & z, & z', & s & ), & & \Cl_\gamma(& x_{u_1}, & f, & t, & s, & s', & s & ), \\
\Cl_\gamma(& x_{u_1}, & x_{\neg u_2}, & f, & z, & z', & s & ), & & \Cl_\gamma(& x_{u_1}, & x_{\neg u_2}, & f, & s, & s', & s & ), \\
\Cl_\gamma(& x_{u_1}, & x_{\neg u_2}, & t, & z, & z', & s & ), & & \Cl_\gamma(& x_{u_1}, & x_{\neg u_2}, & t, & s, & s', & s & ).
\end{array}
\label{eq_ex}
\end{equation}

Finally, let $g$ be the first atom above, $\R(z, z', y, y')$.

\smallskip

We prove the reduction correct by means of the following two claims, the first for the forward direction and the second for the backward direction.

\begin{claim}
If $\tau$ is not critical for $Q$ relative to $g$ then $\phi$ is valid.
\end{claim}

\begin{proof}  
Let $\tau$ be not critical for $Q$ relative to $g$. We need to prove that $\phi$ is valid. 
To this end, consider an arbitrary assignment $\sigma$ of $\bar u$. We show that there exists an extension of $\sigma$ to $\bar v$ such that $\psi$ evaluates to $\mathtt{true}$ under the extended $\sigma$.

Consider a mapping $h$ sending 
\begin{itemize}
\item[--] $z$, $z'$, $y$, and $y'$ to $0$, $0$, $1$, and $2$, respectively; 
\item[--] $x_u$, $y_u$, and $y'_u$ to $0$, $1$, and $2$, respectively, for each $u \in \bar u$ such that $\sigma(u) = \mathtt{false}$;
\item[--] $x_{\neg u}$, $y'_{\neg u}$, and $y_u$ to $0$, $1$, and $2$, respectively, for each $u \in \bar u$ such that $\sigma(u) = \mathtt{true}$; and
\item[--] all other variables to fresh constants. 
\end{itemize}
Let also $I = h(Q)$, so that $h$ is a homomorphism from $Q$ to $I$ sending $g$ to $\tau$. In particular, $I$ contains the following tuples over $\R$:
$$
\begin{array}{ll}
\R(0, 0, 1, 2), & \text{(i.e., } \tau \text{)}, \\
\R(h(s), h(s'), h(p), h(p)), & \\
\R(h(x_{\neg u}), h(x_{\neg u}), h(y'_{\neg u}), 1), & \text{for each }  u \in \bar u, \text{ such that } \sigma(u) = \mathtt{false},\\
\R(h(x_u), h(x_u), 2, h(y'_u)), & \text{for each }  u \in \bar u, \text{ such that } \sigma(u) = \mathtt{true},\\
\R(h(f), h(f), 1, 1), & \text{if there is }  u \in \bar u, \text{ such that } \sigma(u) = \mathtt{false}, \\ 
\R(h(f), h(f), 2, 2), & \text{if there is }  u \in \bar u, \text{ such that } \sigma(u) = \mathtt{true}.
\end{array}
$$

Since, by our assumption, $\tau$ is not critical for $Q$ relative to $g$, there exists a homomorphism $h_{\new}$ from $Q$ to $I \setminus \{\tau\}$. 

First, we claim that $h_{\new}(z) = h(s)$ (which is different from $h(z) = 0$ by the fact that $h(s)$ is a fresh constant). Indeed, on the one hand, $h_\new(\R(z, z', y, y'))$ should be one of the $\R$ tuples in $I$ except $\tau$, so $h_{\new}(z)$ is one of $h(s)$, $h(x_{\neg u})$, $h(x_u)$ (for $u \in \bar u$), and $h(f)$. On the other hand, $h_\new(\Cl_\gamma(\ldots, r, r, z))$, for any clause $\gamma$, is the image of one of the $\Cl_\gamma$ atoms under $h$, so $h_{\new}(z)$ is either $h(z) = 0$ or $h(s)$. Therefore, the only option satisfying both conditions is $h_{\new}(z) = h(s)$. This fact also implies that $h_{\new}(z') = h(s')$. Moreover, since $h(s) \neq h(s')$, $h_\new(r) = h(z) = h(z') = 0$.

We also claim that $h_{\new}(x_v)$ is either $h(f)$ or $h(t)$ for each $v \in \bar v$. Indeed, since $h_{\new}(z) = h(s) \neq h(z) = 0$ and $h_\new(r) = h(z) = 0 \neq h(s)$, $h_\new(\Cl_\gamma(\ldots, r, r, z))$ is the image of one of the 7 $\Cl_\gamma$ atoms under $h$ for any $\gamma$. So, considering any $\gamma$ that has a literal over any variable $v \in \bar v$, $h_{\new}(x_v)$ is either $h(f)$ or $h(t)$.

Therefore, we can consider the extension of assignment $\sigma$ to $\bar v$ such that, for any $v \in \bar v$, $\sigma(v) = \mathtt{false}$ if $h_{\new}(x_v) = h(f)$ and $\sigma(v) = \mathtt{true}$ otherwise. We next prove that $\psi$ evaluates to $\mathtt{true}$ under the extended $\sigma$.

To this end, we first show that $h_{\new}(x_u) = h(f)$ for each $u \in \bar u$ such that $\sigma(u) = \mathtt{false}$. Indeed, this can be shown similarly to the claim above about $h_{\new}(z) = h(s)$. On the one hand, $h_\new(\R(x_u, x_u, y_u, y'_u))$ should be one of the $\R$ tuples in $I$ with the first two arguments equal and except $\tau$, so $h_{\new}(x_u)$ is one of $h(x_{\neg u})$, $h(x_{u'})$, $h(x_{\neg u'})$ (for $u' \neq u$), and $h(f)$. On the other, by our assumption there exists a clause $\gamma$ in $\psi$ with a positive literal $u$, so, for this $\gamma$, $h_\new(\Cl_\gamma(\ldots, x_{u}, \ldots, r, r, z))$ is the image of one of the $\Cl_\gamma$ atoms under $h$, which implies that $h_{\new}(x_u)$ is either $h(f)$ or $h(x_u)$. Therefore, the only option satisfying both conditions is $h_{\new}(x_u) = h(f)$. Moreover, this only option for the $\R$ atom implies that $h_\new(y_u) = h_\new(y'_u) = h(y_u) = 1$.

In fact, $h_{\new}(x_{\neg u})$ is either $h(x_{\neg u})$ or $h(f)$ for each $u \in \bar u$ as above---that is, such that $\sigma(u) = \mathtt{false}$. This can be shown in exactly the same way as above, except that $h(x_{\neg u})$ is also an option in the first condition, because $h(\R(x_{\neg u}, x_{\neg u}, y'_{\neg u}, y_u))$ is not $\tau$ and hence belongs to $I \setminus \{\tau\}$. Moreover, this implies that if $h_{\new}(x_{\neg u}) = h(x_{\neg u})$ then $h_{\new}(y'_{\neg u}) = h(y'_{\neg u})$, and if $h_{\new}(x_{\neg u}) = h(f)$ then $h_{\new}(y'_{\neg u}) = h(y_{u}) = 1$.

In a symmetrical way, we can prove that, for each $u \in \bar u$ such that $\sigma(u) = \mathtt{true}$, $h_{\new}(x_{\neg u}, y'_{\neg u}, y_u)$ is $h(f), 2, 2$ and $h_{\new}(x_u, y'_u)$ is either $h(x_u, y'_u)$ or $h(f), 2$.

Consider now the mapping $h'_{\new}$ that is the same as $h_{\new}$ except that
\begin{itemize}
\item[--] $h'_{\new}(x_{\neg u}, y'_{\neg u}) = h(x_{\neg u}, y'_{\neg u})$ for each $u \in \bar u$ with $\sigma(u) = \mathtt{false}$, and
\item[--] $h'_{\new}(x_u, y'_u) = h(x_u, y'_u)$ for each $u \in \bar u$ with $\sigma(u) = \mathtt{true}$.
\end{itemize}
In other words, if $h_\new$ has two options as described above for some $x_{\neg u}$ or $x_u$, then $h'_{\new}$ takes the option that is different from $h(f)$. 

We now claim that $h'_{\new}$ is also a homomorphism from $Q$ to $I \setminus \{\tau\}$. To verify this, we need to check that the images of all the atoms with $x_{\neg u}$ and $y'_{\neg u}$ in $Q$ under $h'_{\new}$ are in $I \setminus \{\tau\}$ for every $u \in \bar u$ with $\sigma(u) = \mathtt{false}$, and the images of all the atoms with $x_u$ and $y'_y$ are in $I \setminus \{\tau\}$ for every $u$ with $\sigma(u) = \mathtt{true}$. We do it just for the first case, and the second is symmetric. To this end, consider any $u \in \bar u$ with $\sigma(u) = \mathtt{false}$.

First, $x_{\neg u}$ and $y'_{\neg u}$ appear in the atom $\R(x_{\neg u}, x_{\neg u}, y'_{\neg u}, y_u)$ in $Q$. As we discussed above, there are two options for $h_{\new}(\R(x_{\neg u}, x_{\neg u}, y'_{\neg u}, y_u))$: it is either $\R(h(x_{\neg u}), h(x_{\neg u}), h(y'_{\neg u}), 1)$ or $\R(h(f), h(f), 1, 1)$. By construction, $h'_{\new}$ just picks the first option for this atom.

Second, $x_{\neg u}$ appears in atom~\eqref{eq_clauses} over $\Cl_\gamma$ for each $\gamma$ with $\neg u$ (for our example clause $u_1 \lor \neg u_2 \lor \neg v$ and $u_2$ this atom is listed first in~\eqref{eq_ex}). Since $h_\new(z) = h(s)$ and $h(s) \neq h(s')$, the image of this $\Cl_\gamma$ atom under $h_\new$ is the image of one of the 7 atoms~\eqref{eq_assignments} for $\gamma$ under $h$, corresponding to the assignments satisfying $\gamma$ (in the example they are in the left column). Here, a change of the value of $x_{\neg u}$ from $h(f)$, as can be in $h_\new$, to $h(x_{\neg u})$, as in $h'_\new$, preserves a homomorphism, because a change of value of a literal in a clause from $\mathtt{false}$ to $\mathtt{true}$ cannot turn the assignment from satisfying to not satisfying (in the example, for every atom with $f$ as the second argument in the first column there is the same atom with $x_{\neg u_2}$ instead).

Third, $x_{\neg u}$ appears in several atoms~\eqref{eq_assignments} over $\Cl_\gamma$ for each $\gamma$ as in the previous case (e.g., the 4 atoms with $x_{\neg u_2}$ in the rest of the first column in~\eqref{eq_ex}). Since $h_\new(z) = h(s)$ and $h_\new(z') = h(s')$, the images of these $\Cl_\gamma$ atoms under $h_\new$ are among the images of the 8 atoms~\eqref{eq_backups} for $\gamma$ under $h$ (i.e., the right column in~\eqref{eq_ex} for the example). Possible change from $h(f)$ to $h(x_{\neg u})$ as the value of $x_{\neg u}$ preserves the homomorphism by the construction of atoms~\eqref{eq_backups}.

Finally, $x_{\neg u}$ appears in several atoms~\eqref{eq_backups} over $\Cl_\gamma$ for each $\gamma$ (e.g., the 4 atoms with $x_{\neg u_2}$ in the second column in~\eqref{eq_ex}). Since $h_\new(s) = h(s)$, the images of these $\Cl_\gamma$ atoms under $h_\new$ are also among~\eqref{eq_backups}, and changing the value of $x_{\neg u}$ from $h(f)$ to $h(x_{\neg u})$ preserves the homomorphism as in the previous case.

So, overall we have that $h'_{\new}$ is a homomorphism from $Q$ to $I \setminus \{\tau\}$ such that
\begin{itemize}
\item[--] $h'_{\new}(x_u) = h(f)$ and $h'_{\new}(x_{\neg u}) = h(x_{\neg u})$ for each $u \in \bar u$ with $\sigma(u) = \mathtt{false}$,
\item[--] $h'_{\new}(x_u) = h(x_u)$ and $h'_{\new}(x_{\neg u}) = h(f)$ for each $u \in \bar u$ with $\sigma(u) = \mathtt{true}$,
\item[--] $h'_{\new}(x_v) = h(f)$ for each $v \in \bar v$ with $\sigma(v) = \mathtt{false}$, and
\item[--] $h'_{\new}(x_v) = h(t)$ for each $v \in \bar v$ with $\sigma(v) = \mathtt{true}$.
\end{itemize}
By the construction of atoms~\eqref{eq_assignments} this means that $\sigma$ satisfies all clauses $\gamma$---that is, $\psi$ is $\texttt{true}$ under $\sigma$. Since our original choice of $\sigma$ as an assignment of $\bar u$ was arbitrary, $\phi$ is indeed valid, as required.
\end{proof}

\begin{claim}
If $\phi$ is valid then $\tau$ is not critical for $Q$ relative to $g$.
\end{claim}

\begin{proof}
Let $\phi$ is valid---that is, for all assignments of $\bar u$ there exists its extension to $\bar v$ such that $\psi$ evaluates to $\mathtt{true}$ under the extention. 
We need to show that $\tau$ is not critical for $Q$ relative to $g$. To this end, consider an arbitrary instance $I$ with a homomorphism $h$ from $Q$ to $I$ sending $g$ to $\tau$.
We prove that there is a homomorphism from $Q$ to $I \setminus \{\tau\}$. 
Consider an assignment $\sigma$ of $\bar u \cup \bar v$ satisfying $\psi$ such that $\sigma(u) = \mathtt{false}$ for $u \in \bar u$ if and only if $h$ sends the atom $\R(x_u, x_u, y_u, y'_u)$ of $Q$ to $\tau$ (which exists since $\phi$ is valid). We construct a mapping $h_{\new}$ from the variables of $Q$ to constants so that it sends
\begin{itemize}
\item[--] $z$, $z'$, $y$, and $y'$ to $h(s)$, $h(s')$, $h(p)$, and $h(p)$, respectively; 
\item[--] $x_u$ and $y'_u$ to $h(f)$ and $h(y_u) = 1$, respectively, for every  $u \in \bar u$ with $\sigma(u) = \mathtt{false}$;
\item[--] $x_{\neg u}$ and $y'_{\neg u}$to $h(f)$ and $h(y_u) = 2$, respectively, for every  $u \in \bar u$ with $\sigma(u) = \mathtt{true}$; 
\item[--] $x_v$ to $h(f)$ for every $v \in \bar v$ with $\sigma(v) = \mathtt{false}$;
\item[--] $x_v$ to $h(t)$ for every $v \in \bar v$ with $\sigma(v) = \mathtt{true}$;
\item[--] $r$ to $h(z) = h(z') = 0$; and
\item[--] all other variables same as $h$.
\end{itemize}

%

We claim that $h_{\new}$ is a homomorphism from $Q$ to $I \setminus \{\tau\}$. Indeed, $h_{\new}$ sends 
\begin{itemize}
\item[--] atoms $\R(z, z', y, y')$ and $\R(s, s', p, p)$ to $h(\R(s, s', p, p))$; 
\item[--] for each $u \in \bar u$, atom $\R(x_u, x_u, y_u, y'_u)$ either to $h(\R(f, f, y_u, y_u))$, if $\sigma(u) = \texttt{false}$, or to $h(\R(x_u, x_u, y_u, y'_u))$ otherwise; 
\item[--] for each $u \in \bar u$, atom $\R(x_{\neg u}, x_{\neg u}, y'_u, y_u)$ either to $h(\R(f, f, y_u, y_u))$, if $\sigma(u) = \texttt{true}$, or to $h(\R(x_{\neg u}, x_{\neg u}, y'_u, y_u))$ otherwise;
\item[--] for each $u \in \bar u$, atom $\R(f, f, y_u, y_u)$ to $h(\R(f, f, y_u, y_u))$; 
\item[--] each atom~\eqref{eq_clauses} to the $h$-image of the atom in~\eqref{eq_assignments} corresponding to the projection of $\sigma$ to the variables of the clause;
\item[--] each atom in~\eqref{eq_assignments} to an atom in~\eqref{eq_backups}; and
\item[--] each atom in~\eqref{eq_backups} to a (possibly different) atom in~\eqref{eq_backups} as well.
\end{itemize}

Since the choice of $I$ was arbitrary, we conclude that $\tau$ is indeed not critical for $Q$ relative to $g$.
\end{proof}

Therefore, the reduction is indeed correct and the theorem is proven.
\end{proof}

\section{Complexity of critical tuple problem}
In this section we show NP-hardness of the original critical tuple problem. Since the best known upper bound for this problem is $\Pi^p_2$ (see \cite{DBLP:journals/jcss/MiklauS07}, the proof is very similar to the membership proof in Theorem~\ref{thm_rct}), we leave the precise complexity open.

\begin{theorem}
Given a tuple $\tau$ and a query $Q$, it is \textnormal{NP}-hard to decide whether $\tau$ is not critical for $Q$.
\label{thm_ct}
\end{theorem}
\begin{proof}
The proof is by reduction of the graph homomorphism problem. Let $G_1$ and $G_2$ be two directed graphs. Without loss of generality, we assume that $G_1$ is \emph{weakly connected}---that is, there exists an undirected path from any node to any node. We construct a tuple $\tau$ and a query $Q$, and then prove that $\tau$ is not critical for $Q$ if and only if there is a homomorphism from $G_1$ to $G_2$.

Let $\tau$ be $\R(0, 1)$ with $0$ and $1$ constants, while $Q$ consist of the following atoms, where all arguments are variables:
\begin{itemize}
\item[--] $\E(x_{v_1}, x_{v_2})$ for each edge $(v_1, v_2)$ of $G_1$ and $G_2$;
\item[--] $\R(x_{v^*}, x)$ for an arbitrarily chosen node $v^*$ of $G_1$; and
\item[--] $\R(x_v, x_v)$ for each node $v$ of $G_2$. 
\end{itemize}
  
We next prove that $\tau$ is not critical for $Q$ if and only if there is a homomorphism from $G_1$ to $G_2$, and start with the forward direction. Since $\tau$ is not critical for $Q$, for every instance $I$ with a homomorphism from $Q$  there is a homomorphism from $Q$ to $I \setminus \{\tau\}$. Consider a mapping $h$ from the variables of $Q$ to constants that sends $x_{v^*}$ and $x$ to $0$ and $1$, respectively, and all other variables to fresh constants. Let $I = h(Q)$, so that $h$ is a homomorphism from $Q$ to $I$. Since $\tau$ is not critical, there is a homomorphism $h_{\new}$ from $Q$ to $I \setminus \{\tau\}$. Note that, by construction, there is no node $v$ of $G_1$ and constant $a$ such that $\R(h(x_v), a)$ is in $I \setminus \{\tau\}$. Therefore, $h_{\new}(x_{v^*})$ has to be $h(x_v)$ for a node $v$ of $G_2$. Since $G_1$ is weakly connected and graphs $G_1$ and $G_2$ do not have common nodes, this holds also for all other nodes of $G_1$. So, $h_{\new}$ defines a homomorphism from $G_1$ to $G_2$, as required.

We proceed to prove the backward direction of the claim. To this end, let there be a homomorphism $h'$ from $G_1$ to $G_2$. We need to prove that $\tau$ is not critical for $Q$---that is, for every instance $I$ with a homomorphism from $Q$  there is a homomorphism from $Q$ to $I \setminus \{\tau\}$. Consider such an instance $I$ and a homomorphism $h$ from $Q$ to $I$. Since for any node $v$ of $G_2$ the atom $\R(x_v, x_v)$ cannot be mapped to $\tau$ by $h$, the mapping $h_{\new}$ that sends
\begin{itemize}
\item[--] $x_v$ to $h(x_{h'(v)})$ for each node $v$ of $G_1$,
\item[--] $x$ to $h(x_{h'(v^*)})$, and 
\item[--] $x_v$ to $h(x_v)$ for each node $v$ of $G_2$
\end{itemize}
is a homomorphism from $Q$ to $I \setminus \{\tau\}$, as required.
\end{proof}

\section{Conclusion}
Besides its own importance in database theory, the $\Pi_2^p$-hardness of the critical tuple problem, claimed in~\cite{DBLP:journals/jcss/MiklauS07}, has been used to establish $\Pi_2^p$-hardness of several other problems throughout the literature (e.g., in \cite{DBLP:conf/aaai/GrauK16}). Therefore, until the $\Pi_2^p$-hardness of the critical tuple problem is proved, all these results should be revised. In particular, it should be identified which of these problems could be proven $\Pi_2^p$-hard by an independent reduction and which are for now open together with the critical tuple problem.

\bibliographystyle{plain}
\bibliography{bib}

\end{document}